\documentclass{birkjour}

\usepackage{amsmath,amssymb,bm}
\usepackage{amsfonts}
\usepackage{psfrag,graphicx}
\usepackage{color}
\usepackage{colordvi}
\newtheorem{theorem}{Theorem}

 \numberwithin{equation}{section}

\def\bea{\begin{eqnarray*}}
\def\eea{\end{eqnarray*}}
\def\disp{\displaystyle}

\def\bq{{\bf q}}
\def\bv{{\bf v}}
\def\bn{{\bf n}}

\begin{document}

%
%
%
%
%
%
%
%
%

\title[Phase separation in quasi incompressible fluids]
 {Phase separation in quasi incompressible fluids: Cahn--Hilliard model in the Cattaneo--Maxwell
framework}

\author[A. Berti]{Berti Alessia}
\address{Faculty of Engineering\\ 
University e-Campus\\
Via Isimbardi 10\\ 
22060 Novedrate (CO), Italy.}
\email{alessia.berti@ing.unibs.it}%
\author[I. Bochicchio]{Ivana Bochicchio}%
\address{Department of Mathematics\br
University of Salerno \br
Via Ponte Don Melillo\br 
84084 Fisciano (SA), Italy.}
\email{ibochicchio@unisa.it}
\author[M. Fabrizio]{Mauro Fabrizio}%
\address{Department of Mathematics\\
University of Bologna\\
Piazza Porta San Donato 5\\ 
40126 Bologna, Italy.}
\email{fabrizio@dm.unibo.it}



\date{ }

\begin{abstract}
In this paper we propose a mathematical model of phase separation for a quasi-incompressible binary mixture where the spinodal decomposition is induced by an heat flux governed by the Cattaneo--Maxwell equation. 
As usual, the phase separation is considered in the framework of phase field modeling so that the transition is described by an additional field, the concentration $c$.
The evolution of concentration is described by the Cahn--Hilliard equation and in our model is coupled with the Navier--Stokes equation.
Since thermal effect are included, the whole set of evolution equations is set up for the velocity, the concentration, the temperature and the heat flux. The model is compatible with thermodynamics and a maximum theorem holds.
\end{abstract}

\maketitle

\section{Introduction}
The mechanism by which a mixture of two or more components can spontaneously separate into distinct regions (or phases) with different chemical compositions and physical properties is usually named spinodal decomposition or phase separation.
This phenomenon has been widely studied with phase field approach (see for instance \cite{CH,FGM,FGM_physicaD,gurtin,LT} and reference therein), in that the interface between the two pure phases is not sharp and it is replaced by a narrow diffuse
layer across which the fluids may mix.
If we denote with $c$ the concentration of the components, its evolution is given by the Cahn-Hilliard equation:
$$
\dot c = \nabla \cdot [M(c) \nabla \mu] \ ,
$$
where $M(c)$ represents the mobility and $\mu$ is the chemical potential depending on the state variables.

The phase separation can be induced by many factors. Typically it takes place when the mixture is quickly cooled below a critical value of the temperature where the mixture can no longer exist in equilibrium in its homogeneous state (\cite{FGM}). 
Even the velocity can influence the miscibility properties of the mixture (see \cite{BBG,BB,FGM,LT}).

In our paper, we suppose that the heat flux $\bq$ can induce the spinodal decomposition. Indeed, an increase of $\bq$, like as an increase in the temperature, reduce the miscibility gap. So we let the chemical potential depend on $\bq$. 
In order to describe the evolution of the system, we couple the kinetic equations involving the state variables with a suitable law for the heat flux. In particular, we assume that $\bq$ obeys a (modified) Cattaneo-Maxwell equation (see \cite{cattaneo,CJ,F1,F2}).
It plays a crucial role in proving the thermodynamically consistence of our model, which is not guaranteed with a constitutive law of Fourier type.

The paper is organized as follows. In Section 2 we introduce the order parameter $c$ and, following \cite{LT}, we model the system as a quasi-incompressible binary mixture. The assumption of quasi-incompressibility means that both components are incompressible with different density, but, due to variations of the order parameter, the density of the mixture is not constant and the velocity may not be non-solenoidal. 
In Section 3 we write the evolution equations for the state variable (the order parameter, the velocity, the absolute temperature and the heat flux). Section 4 is devoted to establish the restrictions imposed on the material parameters by the principles of thermodynamics. Finally, in Section 5 we prove a maximum theorem for the order parameter, so that $c$ is always defined into the interval $[-1,1]$.

\section{Preliminaries}
We consider a binary mixture of two incompressible non-reacting fluids, occupying a fixed domain $\Omega \subset \mathbb{R}^{3}$ with a smooth boundary
$\partial \Omega$. In the following the fluids are labeled by $i=1,2$. 
Each component is characterized by its own intrinsic constant density $\rho_{i0}$ under standard conditions of temperature and pressure. We suppose that $\rho_{10} \neq \rho_{20}$.

The total mass and the density of the mixture are denoted respectively by $M$ and $\rho $, namely 
\begin{equation*}
M=\int_{\Omega }\rho \,dx.
\end{equation*}
Let $M_{1}$, $M_{2}$ be the masses of each species in $\Omega $, so that $
M=M_{1}+M_{2}$. We denote by $\rho _{1}$ and $\rho _{2}$ the apparent
densities of the two constituents, such that $\rho =\rho _{1}+\rho _{2}$. The adjective ``apparent'' is used to emphasize that we are considering the ratio of each mass
fraction over the total volume element, rather than over its own fractional volume. Accordingly, the ratio $\rho_i/\rho_{i0}$ denotes the the volume fraction of the substance $i$ and hence the following equality holds:
\begin{equation}\label{vol_frac}
	\frac{\rho_1}{\rho_{10}} + \frac{\rho_2}{\rho_{20}} =1.
\end{equation}

Denoting by $\bv_i$ the velocity of the $i$ fluid, the mean velocity $\mathbf{v}$ is defined by 
\begin{equation*}
\rho\mathbf{v}= \rho_1\mathbf{v}_{1}+\rho_{2}\mathbf{v}_{2}.
\end{equation*}

In order to derive the diffuse interface model, we introduce an order parameter measuring the degree of phase separation, {\it e.g.} 
\begin{equation*}
c= c_1 -c_2 =\frac{\rho _{1}-\rho _{2}}{\rho},
\end{equation*}
where $c_i = \rho_i/\rho$ denotes the mass concentration of the fluid $i$. The equality $\rho= \rho_1+\rho_2$ leads to
\begin{equation}\label{c12}
	c_1 = \frac{1+c}{2}, \quad c_2 = \frac{1-c}{2}.
\end{equation}
From the definition of $c$, it is apparent that $c\in \lbrack -1,1]$. In particular, $c=-1$
(or $c=1$) wherever only the component $1$ (or $2$) occurs. In contrast with two fluids models, in the diffusive approach the fundamental fields of the model are $\rho$, $\bv$, $c$, rather than $\rho_1$, $\rho_2$, $\bv_1$, $\bv_2$ (see \cite{morro}). 

In our paper, we are interested in modeling quasi-incompressible fluids, that is we assume that both constituents are incompressible, but the density of the mixture may not be constant and change owing to variations in the concentration parameter $c$. For this reason, the density $\rho $ is no more an
independent variable, but it is a function of $c$. In particular, \eqref{vol_frac} and \eqref{c12} implies
\begin{equation*}
\frac{1}{\rho }=\frac{1+c}{2}\cdot \frac{1}{\rho _{10}}+\frac{1-c}{2}\cdot 
\frac{1}{\rho _{20}},
\end{equation*}
which implies 
\begin{equation}
\rho \,=\,\frac{2\rho _{10}\rho _{20}}{(\rho _{10}+\rho _{20})-c(\rho_{10}-\rho _{20})}.  \label{ro}
\end{equation}
The assumption $\rho_{10} \neq \rho _{20}$ assures that the density is not constant. Accordingly, the velocity $\mathbf{v}$ is not solenoidal and satisfies the
continuity equation 
\begin{equation}
\dot{\rho}=-\rho \nabla \cdot \mathbf{{v}.}  \label{balance1}
\end{equation}
Since $\rho $ is a function only of $c$, $\dot{c}$ is related to $\nabla \cdot \mathbf{v}$ by the relation
\begin{equation} 
\rho_{c}\,\dot{c}=-\rho \nabla \cdot \mathbf{{v}.}  \label{con}
\end{equation}
From \eqref{ro} it follows that the derivative $\rho_{c}$ is given
by 
\begin{equation} \label{ro-der}
\rho_{c} =
\frac{2\rho _{10}\rho _{20}(\rho _{10}-\rho _{20})}{[(\rho _{10}+\rho _{20})-c(\rho _{10}-\rho _{20})]^{2}}.
\end{equation}

During the process of phase separation, the two components can separate into distinct regions with different chemical compositions, but the total mass $M_1$, $M_2$, of the two species remain constant, that is
\begin{equation}\label{balance_12}
	\partial_t\rho_i + \nabla \cdot(\rho_i\bv_i)=0, \quad i=1,2.
\end{equation}
Equations \eqref{balance_12} are equivalent, in a two fluids model, to the balance of the overall mass \eqref{balance1} and 
the balance of the order parameter
\begin{equation}
\rho \dot{c}\,=\,\nabla \cdot \mathbf{j}\,,  \label{(2.6)}
\end{equation}
where the vector $\mathbf{j}$ is a suitable flux (see \cite{FGM,morro}) satisfying the boundary condition
\begin{equation}
\mathbf{j}\cdot \mathbf{n}=0 \quad \text{at}\quad \partial \Omega,
\label{(2.8)}
\end{equation}
where $\bn$ denotes the unit outward normal vector.
As a consequence, the global mass of the mixture is conserved, namely
\begin{equation}
\frac{d}{dt}\int_{\Omega }\rho c\,dv = \int_{\Omega }\rho \dot c\,dv = \int_{\partial\Omega }\mathbf{j}\cdot \mathbf{n}\,da =0.  \label{(2.5)}
\end{equation}

As customary, we regard $\mathbf{j}$ as a constitutive function of $\rho$, $\bv$, $c$ (and their gradients).
In particular, $\mathbf{j}$ is assumed to be proportional to the gradient of the generalized chemical potential $\mu $, \textit{i.e.} 
\begin{equation*}
\mathbf{j}\,=\,M(c)\nabla \mu
\end{equation*}
where $M(c)$ denotes the diffusive mobility, which is a non-negative function eventually depending on the concentration $c$, while $\mu $ is the
classical chemical potential.



\section{Evolution equations}

This section is devoted to recall the evolution equations for the fields of
our model by the following balance equations 
\begin{eqnarray}
\rho \dot{\mathbf{v}} &=&\nabla \cdot \mathbf{T}+\rho \mathbf{b}  \label{3.1}
\\
\rho \dot{c} &=&\nabla \cdot \lbrack M(c)\nabla \mu ]  \label{3.2}
\end{eqnarray}%
where $\mathbf{T}$ is the stress tensor which depends on the symmetrical
part $\mathbf{D}$ of the gradient of velocity, the concentration $c$ and its
gradient $\nabla c$, while $\mathbf{b}$ denotes the body force density. So
that, we assume that $\mathbf{T}$ is given by the sum of two second-order
tensors, \textit{i.e.} 
\begin{equation*}
\mathbf{T}(\mathbf{D},c,\nabla c)=\mathbf{T}_{1}(\mathbf{D},c)+\mathbf{T}%
_{2}(c,\nabla c),\qquad \mathbf{D}=\frac{1}{2}[\nabla \mathbf{v}+(\nabla 
\mathbf{v})^{T}].
\end{equation*}
The first term is related to the classical Cauchy stress tensor for a
viscous fluid, that is 
\begin{equation*}
\mathbf{T}_{1}(\mathbf{D},p,\rho ,c)=-p(c)\mathbf{1}+2\nu {(c)}\mathbf{D}%
+\sigma (c)(\nabla \cdot \mathbf{v})\mathbf{1},
\end{equation*}
where $\mathbf{1}$ stands for the second-order identity tensor, $\nu (c)$
and $\sigma (c)$ denote the viscosity coefficients of the mixture. In
particular, when $c=1$ (or $c=-1$) $\nu $ and $\sigma $ coincide with the
viscosity of the fluid 1 (or 2). Here, since the density $\rho ~$depends
on $c$, we let the pression $p$ be a function of $c$. This view point is
well described in \cite{FGM}, where it is shown the relevant changes which
occur if, as in \cite{LT}, $p$ is regarded as a unknown function.

The tensor $\mathbf{T}_{2}$ accounts for the capillary forces due to surface
tension and it is associated to the gradient of the concentration (see 
\textit{e.g.} \cite{gurtin}), \textit{i.e.} 
\begin{equation*}
\mathbf{T}(c,\nabla c)=-\gamma \rho (c)\nabla c\otimes \nabla c,
\end{equation*}%
where the parameter $\gamma $ is related to the thickness of the interfacial
region.

As a consequence, the stress tensor is given by 
\begin{equation*}
\mathbf{T}=-p(c)\mathbf{1}-\gamma \rho \nabla c\otimes \nabla c+2\nu {(c)}%
\mathbf{D}+\sigma (c)(\nabla \cdot \mathbf{v})\mathbf{1}
\end{equation*}
and the linear momentum balance equation reads 
\begin{equation}
\rho \dot{\mathbf{v}}=-\nabla p(c)-\gamma \nabla \cdot (\rho \,\nabla
c\otimes \nabla c)+\nabla \cdot (2\nu (c)\mathbf{D})+\nabla (\sigma
(c)\nabla \cdot \mathbf{v})+\rho \mathbf{b}.  \label{v}
\end{equation}

Now we focus our attention on the diffusion equation 
\begin{equation*}
\rho \dot{c}=\nabla \cdot \lbrack M(c)\nabla \mu ].
\end{equation*}%
Here, we consider a generalization of the chemical potential by assuming
that $\mu $ depends by the concentration $c$ and its gradient $\nabla c$,
the absolute temperature $\theta $ and the heat flux $\mathbf{q}$. The
underlying physical idea is that heat flux can influence the miscibility
properties of the mixture, namely an increase in the heat flux (like an
increase in the temperature) improves the miscibility of the mixture.
Accordingly, we suppose that $\mu $ is defined as 
\begin{equation}\label{mi}
\mu =-{\frac{\gamma }{\rho }}\nabla \cdot (\rho \nabla c)+\,\theta
_{0}F_c(c)+\,\,\left[ \theta +\frac{1}{\kappa _{0}}|\mathbf{q}|^{2}%
\right] G_c(c),  
\end{equation}
where $\theta _{0},\kappa _{0}$ are positive constants and $F,G$ are
suitable functions depending only on $c$ and whose expression will be given in the sequel.
With this choice, the evolution equation for the concentration is given by 
\begin{equation}\label{c}
\rho \dot{c}
=
\nabla \cdot \left[ M(c)\nabla \left( -\frac{\gamma }{\rho} \nabla \cdot (\rho \nabla c)+\,\theta _{0}F_c(c)+\,\,\left[ \theta +
\frac{1}{\kappa _{0}}|\mathbf{q}|^{2}\right] G_c(c)\right) \right]
\end{equation}

In order to obtain the equation for the temperature, let us consider the
first law of thermodynamics in the form 
\begin{equation}
\rho \dot{e}={\mathcal{P}}_{m}^{i}+{\mathcal{P}}_{c}^{i}+{\mathcal{P}}%
_{h}^{i},  \label{Ilaw}
\end{equation}%
where $e$ the internal energy, which we suppose function of the variables $%
\theta ,c,\nabla c,\mathbf{q}$, ${\mathcal{P}}_{m}^{i}$ the internal
mechanical power, ${\mathcal{P}}_{c}^{i}$ the internal chemical power, while 
${\mathcal{P}}_{h}^{i}\,=\,\rho h$ is the internal heat power and $h$ is
the rate at which heat is absorbed per unit mass (see for instance \cite{fremond}). Denoting by $T=\frac{1}{2}%
\rho \mathbf{v}^{2}$ the kinetic energy and $E$ the total energy, we write $%
E=T+e$.

By multiplying equation \eqref{v} by $\mathbf{v}$, we obtain the mechanical
power balance, that is 
\begin{equation*}
\rho \dot{T}+{\mathcal{P}}_{m}^{i}={\mathcal{P}}_{m}^{e},
\end{equation*}
with ${\mathcal{P}}_{m}^{e}$ the external mechanical power, defined 
\begin{eqnarray}
{\mathcal{P}}_{m}^{i} =-p\nabla \cdot \mathbf{v}+
\nu(c)\mathbf{D}^{2}
+\sigma (c)(\nabla \cdot \mathbf{v})^{2}+\gamma \rho (\nabla c\otimes
\nabla c)\cdot \nabla \mathbf{v}  \label{Pmi} \\
&&  \notag \\
{\mathcal{P}}_{m}^{e} =\nabla \cdot \lbrack -p\mathbf{v}+2\nu (c)\mathbf{D}%
\mathbf{v}-\gamma \rho (\nabla c\otimes \nabla c)\mathbf{v}+\sigma (c)%
\mathbf{v}\nabla \cdot \mathbf{v}]+\rho \mathbf{b}\cdot \mathbf{v}
\label{Pme}
\end{eqnarray}

Similarly, multiplying equation \eqref{c} by $\rho \dot{c}$, we obtain the
power balance related to the concentration $c$, that is 
\begin{equation*}
{\mathcal{P}}_{c}^{i}={\mathcal{P}}_{c}^{e},
\end{equation*}%
where ${\mathcal{P}}_{c}^{e}$ is the external chemical power, such that 
\begin{eqnarray}
{\mathcal{P}}_{c}^{i} &=&
\rho \theta _{0}\dot{F}(c)+\rho \dot{G}(c)\left[
\theta +\frac{1}{\kappa _{0}}|\mathbf{q}|^{2}\right] +\rho \gamma \nabla
c\cdot \nabla \dot{c}+M(c)|\nabla \mu |^{2},  
\label{chem} 
\\
{\mathcal{P}}_{c}^{e} &=&
\nabla \cdot \lbrack \gamma \rho \dot{c}\nabla
c+M(c)\mu \nabla \mu ].
\end{eqnarray}%
\newline
Adding ${\mathcal{P}}_{m}^{i}$ and ${\mathcal{P}}_{c}^{i}$, we
obtain 
\begin{equation}
\begin{array}{lll}
&  & {\mathcal{P}}_{m}^{i}+{\mathcal{P}}_{c}^{i}  =  \rho \displaystyle\frac{d}{%
dt}\left[ \theta _{0}F(c)+\frac{1}{2}\gamma |\nabla c|^{2}\right] -p\nabla
\cdot \mathbf{v} \\ 
&  &  \\ 
& & +\sigma (c)(\nabla \cdot \mathbf{v})^{2}+\nu (c)\mathbf{D}^{2}+\rho %
\left[ \theta +\frac{1}{\kappa _{0}}|\mathbf{q}|^{2}\right] \dot{G}%
(c)+M(c)|\nabla \mu |^{2}%
\end{array}
\label{sommaP}
\end{equation}%
where we have used the identity 
\begin{equation*}
\dot{\overline{\nabla c}}\cdot \nabla c=\nabla c\cdot \nabla \dot{c}-(\nabla
c\otimes \nabla c)\cdot \nabla \mathbf{v}.
\end{equation*}%
From \eqref{con}, remembering that $\rho$ and $c$ are not independent variables, it follows that 
\begin{equation*}
-p\nabla \cdot \mathbf{v}\,=\,\rho \frac{p}{\rho^2} \dot{\rho}\,=\,\rho \frac{d}{dt}P(\rho)
%
\end{equation*}%
where $P_\rho \,=\,\frac{p}{\rho^2}$.
\newline
Hence, 
\begin{equation}
\begin{array}{lll}
& & \rho \dot{T}+{\mathcal{P}}_{m}^{i}+{\mathcal{P}}_{c}^{i}  =  \rho \displaystyle\frac{d}{%
dt}\left[ \frac{1}{2}\mathbf{v}^{2}+\theta _{0}F(c)+\frac{1}{2}\gamma
|\nabla c|^{2}]+P(\rho)\right] \\ 
&  &  \\ 
&  & +\sigma (c)(\nabla \cdot \mathbf{v})^{2}+\nu (c)\mathbf{D}^{2}+\rho %
\left[ \theta +\frac{1}{\kappa _{0}}|\mathbf{q}|^{2}\right] \dot{G}%
(c)+M(c)|\nabla \mu |^{2}.
\end{array}
\label{sommaP2}
\end{equation}
Moreover, \eqref{sommaP2} suggests to define the internal energy $e$ as 
\begin{equation}
e=e_{0}(\theta )+\theta _{0}F(c)+\frac{1}{2}\gamma |\nabla c|^{2}+P(\rho),  \label{internal_energy}
\end{equation}
where $e_{0}$ is a function depending only on the temperature. Then, the
total energy $E$ is given by 
\begin{equation*}
E=T+e=\frac{1}{2}\mathbf{v}^{2}+e_{0}(\theta )+\theta _{0}F(c)+\frac{1}{%
2}\gamma |\nabla c|^{2}+P(\rho).
\end{equation*}%
A comparison with \eqref{Ilaw} yields 
\begin{equation}
\begin{array}{ll}
\rho h & =\rho e_{0}^{\prime }(\theta )\dot{\theta}-\sigma (c)(\nabla \cdot 
\mathbf{v})^{2}-\nu (c)\mathbf{D}^{2} 
-\rho \left[ \theta + \disp \frac{1}{\kappa _{0}}|\mathbf{q}|^{2}\right] \dot{G}
(c)
\\
&
-M(c)|\nabla \mu |^{2}.
\end{array}
\label{h}
\end{equation}%
As well known (see \textit{e.g.} \cite{fremond}), the thermal balance law is
expressed by the following equation 
\begin{equation}
\rho h=-\nabla \cdot \mathbf{q}+\rho r.  \label{thermal_balance}
\end{equation}

In our model, the equation relating the heat flux $\mathbf{q}$ to the
gradient of the temperature, assumes the form of a generalized
Cattaneo--Maxwell equation 
\begin{equation}  \label{CMLaw}
-2 [\delta + G(c)]\dot{\mathbf{q}}=\mathbf{q} + \kappa(\theta) \nabla\theta
\end{equation}
where $\delta$ is a positive constant and $\kappa(\theta)$ denotes the
thermal conductivity. We suppose that the dependence of $\kappa$ on the
absolute temperature is given by 
\begin{equation*}
\kappa(\theta)\,=\,\frac{\kappa_0}{\theta}.
\end{equation*}
Notice that when $G(c)=0$, we recover the usual Cattaneo-Maxwell equation (\cite{cattaneo}).

Substituting equations \eqref{thermal_balance} and \eqref{CMLaw} into
\eqref{h}, we obtain the kinetic equation for the temperature 
\begin{eqnarray*}
&\rho e_{0}^{\prime }(\theta )\dot{\theta}-\sigma (c)(\nabla \cdot \mathbf{v%
})^{2}-\nu (c)\mathbf{D}^{2}-\rho \left[ \theta + \disp\frac{1}{\kappa _{0}}|%
\mathbf{q}|^{2}\right] \dot{G}(c) -M(c)|\nabla \mu |^{2} = \\
&=
\nabla \cdot \left[2(\delta+G(c)) \dot{q}\right]+\nabla \cdot \left[ \kappa (\theta)\nabla \theta \right] +\rho r.
\end{eqnarray*}
Collecting the previous results, we have the following equations
 
\begin{equation*}
\begin{split}
& \rho \dot{\mathbf{v}} =-\nabla p-\gamma \nabla \cdot (\rho \,\nabla
c\otimes \nabla c)+\nabla \cdot (2\nu (c)\mathbf{D})+\nabla (\sigma
(c)\nabla \cdot \mathbf{v})+\rho \mathbf{b} 
\\
&
\\
&\rho \dot{c} = \nabla \cdot \left[ M(c)\nabla \left( -{\frac{\gamma }{\rho }%
}\nabla \cdot (\rho \nabla c)+\,\theta _{0}F^{\prime }(c)+\,\,\left[ \theta +%
\frac{1}{\kappa _{0}}|\mathbf{q}|^{2}\right] G^{\prime }(c)\right) \right] 
\\
& 
\\
&\rho e_{0}^{\prime }(\theta )\dot{\theta} = \sigma (c)(\nabla \cdot \mathbf{%
v})^{2}+\nu (c)\mathbf{D}^{2}+\rho \left[ \theta +\frac{1}{\kappa _{0}}|%
\mathbf{q}|^{2}\right] \dot{G}(c)+M(c)|\nabla \mu |^{2} +\\
&\qquad \qquad + \nabla \cdot \left[2(\delta+G(c)) \dot{q}\right]+\nabla \cdot \left[ \kappa (\theta
)\nabla \theta \right] +\rho r
\\
&
\\
&-2[\delta +G(c)]\dot{\mathbf{q}}=\mathbf{q}+\kappa (\theta )\nabla \theta 
\end{split}
\end{equation*}%
\\
in the unknowns $\mathbf{v},c,\theta ,\mathbf{q}$, where $p=p(c)$ and $\rho $
is a function of $c$ whose expression is given in \eqref{ro}. To these
equations we append the boundary conditions 
\begin{equation*}
\mathbf{v}=\mathbf{0},\qquad \nabla \theta \cdot \mathbf{n}=0,\qquad \nabla
c\cdot \mathbf{n}=0,\qquad \nabla \mu \cdot \mathbf{n}=0\qquad \text{at}\
\partial \Omega .
\end{equation*}

\section{Thermodynamic restrictions}

In order to prove the thermodynamic consistence with our model, we write the
second law of thermodynamics in the Clausius--Duhem form: 
\begin{equation}
\rho \dot{\eta}\geq \rho \frac{h}{\theta }+\frac{1}{\theta ^{2}}\mathbf{%
q\cdot }\nabla \theta ,  \label{Clausius}
\end{equation}%
where $\eta $ is the entropy function.

We introduce the Helmholtz free energy density $\psi$ defined as 
\begin{equation*}
\psi = e - \theta \eta.
\end{equation*}
We suppose that $\psi$ depends on the variables $\theta, c,\nabla c, \mathbf{%
q}$. Inequality \eqref{Clausius} can be written as 
\begin{equation*}
\rho \dot\psi - \rho \dot e + \rho\dot \theta \eta + \rho h + \frac{1}{\theta%
} \nabla\theta \cdot \mathbf{q} \leq 0.
\end{equation*}

In view of \eqref{internal_energy}-\eqref{h} we have 
\begin{eqnarray}
&&
\rho \dot{\psi}+\rho \eta \dot{\theta}-\rho \frac{d}{dt}\left[ \theta
_{0}F(c)+\frac{1}{2}\gamma |\nabla c|^{2}+P(\rho)+\left[ \theta
+\frac{1}{\kappa _{0}}|\mathbf{q}|^{2}\right] G(c)\right] \leq   \notag
\label{CD} 
\\ 
&&
\sigma (c)(\nabla \cdot \mathbf{v})^{2}+\nu (c)\mathbf{D}^{2}-\rho \left[
\theta +\frac{1}{\kappa _{0}}|\mathbf{q}|^{2}\right] ^{\cdot}G(c)+M(c)|\nabla \mu |^{2} +
\\ \notag
&& -\frac{1}{\theta }\nabla \theta \cdot \mathbf{q}.
\end{eqnarray}
Moreover, from equation \eqref{CMLaw} we obtain 
\begin{equation}
\frac{1}{\theta }\nabla \theta \cdot \mathbf{q}=-\frac{2}{\kappa _{0}}%
[\delta +G(c)]\dot{\mathbf{q}}\cdot \mathbf{q}-\frac{1}{\kappa _{0}}|\mathbf{%
q}|^{2}\ .  \label{CM}
\end{equation}
A substitution into \eqref{CD} leads to 
\begin{eqnarray}
&&\rho \left( \psi_{\theta } +\eta \right) \dot{\theta}+\rho \left[\psi_{c} -\theta _{0}F_c(c)-\left( \theta +\displaystyle%
\frac{1}{\kappa _{0}}|\mathbf{q}|^{2}\right) G_c(c)-\frac{p}{\rho^2}\rho_c\right] \dot{c}  \notag  \label{CD_psi2} 
\\
&&
+\left[\psi_{\mathbf{q}} -\displaystyle\frac{2}{\kappa _{0}}
[\delta +G(c)]\mathbf{q}\right] \dot{\mathbf{q}}
+\rho \left(\psi_{\nabla c} -\gamma \nabla c\right) \cdot \dot{\overline{\nabla c}}
-\sigma (c)(\nabla \cdot \mathbf{v})^{2}  \notag \\
&&-\nu (c)\mathbf{D}^{2}-M(c)|\nabla \mu |^{2}-\frac{1}{\kappa _{0}}
\left\vert \mathbf{q}\right\vert ^{2}\leq 0.  \notag
\end{eqnarray}
In order to satisfy such an inequality, we require that 
\begin{eqnarray*}
\psi_{\theta }  &=&-\eta ,
\qquad 
\psi_{c} =\theta_{0}F_c(c)+\left(\theta +\displaystyle\frac{1}{\kappa_{0}}|\mathbf{q}|^{2}\right) G_c(c)+\frac{p}{\rho^2}\rho_c 
\\
\partial _{\mathbf{q}}\psi  &=&\displaystyle\frac{2}{\kappa _{0}}[\delta +G(c)]\mathbf{q},
\qquad 
\psi_{\nabla c} =\gamma \nabla c
\end{eqnarray*}%
\begin{equation*}
\nu (c),\sigma (c),M(c),\kappa _{0}\geq 0.
\end{equation*}%
Then 
\begin{eqnarray}
& \psi =
\theta _{0}F(c) +\left( \theta +\displaystyle\frac{1}{\kappa _{0}}|
\mathbf{q}|^{2}\right) G(c)+\frac{\gamma }{2}|\nabla c|^{2}+P+
\frac{\delta }{\kappa _{0}}|\mathbf{q}|^{2}+\psi _{0}(\theta ),
\label{free_energy} 
\\
& \notag
\\
& \eta  =-\psi_{\theta } =-G(c)-\psi _{0}^{\prime }(\theta ),
\end{eqnarray}
where $\psi _{0}$ is a suitable function (depending only on $\theta $) which
ensures the validity of the condition $\psi =e-\eta \theta =e+ \psi_{\theta}\theta $. A substitution of \eqref{internal_energy} and %
\eqref{free_energy} leads to the equality 
\begin{equation*}
\psi _{0}(\theta )=e_{0}(\theta )+\psi _{0}^{\prime }(\theta )\theta .
\end{equation*}%
Thus, $\psi _{0}$ is given by 
\begin{equation*}
\psi _{0}=\mathcal{C}\theta -\theta \int \frac{e_{0}(\theta )}{\theta ^{2}} \,d\theta ,
\end{equation*}%
with $\mathcal{C}>0$ and 
\begin{equation*}
\eta =-G(c)-\mathcal{C}+\int \frac{e_{0}(\theta )}{\theta ^{2}} \,d\theta + \frac{e_{0}(\theta )}{\theta }.
\end{equation*}
In particular, if we let $e_{0}=\mathcal{C}\theta $, where $\mathcal{C}$
denotes the specific heat, we recover the standard form of $\psi _{0}$ and $%
\eta $, \textit{i.e.} 
\begin{equation*}
\psi _{0}=\mathcal{C}\theta (1-\ln \theta ),\qquad \eta =-G(c)-\mathcal{C}%
\ln \theta .
\end{equation*}

\section{Maximum principle}

If we like that the Cahn--Hilliard equation describes a natural physical
problem, we have to prove a maximum theorem, namely we have to show that the
evolution equations imply that the concentration $c$ is always defined into
the interval $[-1,1]$. 
\newline
To this aim, remembering that the chemical potential is given by 
\begin{equation}\label{mu}
\mu = -\frac{\gamma }{\rho }\nabla \cdot (\rho \nabla c)+ \,\theta_{0}F_c(c)+\left[ \theta +\frac{1}{\kappa _{0}}|\mathbf{q}|^{2}\right] G_c(c),
\end{equation}
we can define $F$ and $G$ by letting 
\begin{equation}
F(c)=(c^{2}-1)^{2}\ ,\qquad c\in \mathbb{R}  \label{F}
\end{equation}
\begin{equation}
G(c)=\frac{1}{2}\left\{ 
\begin{array}{llll}
c^{2}\qquad  & -1\leq c \leq 1 &  &  
\\[0.5em]
1\qquad  & c<-1 \, \cup \, c>1 &  & 
\end{array}%
\right.   \label{G}
\end{equation}

\begin{figure}[ht]
\includegraphics[scale=0.25]{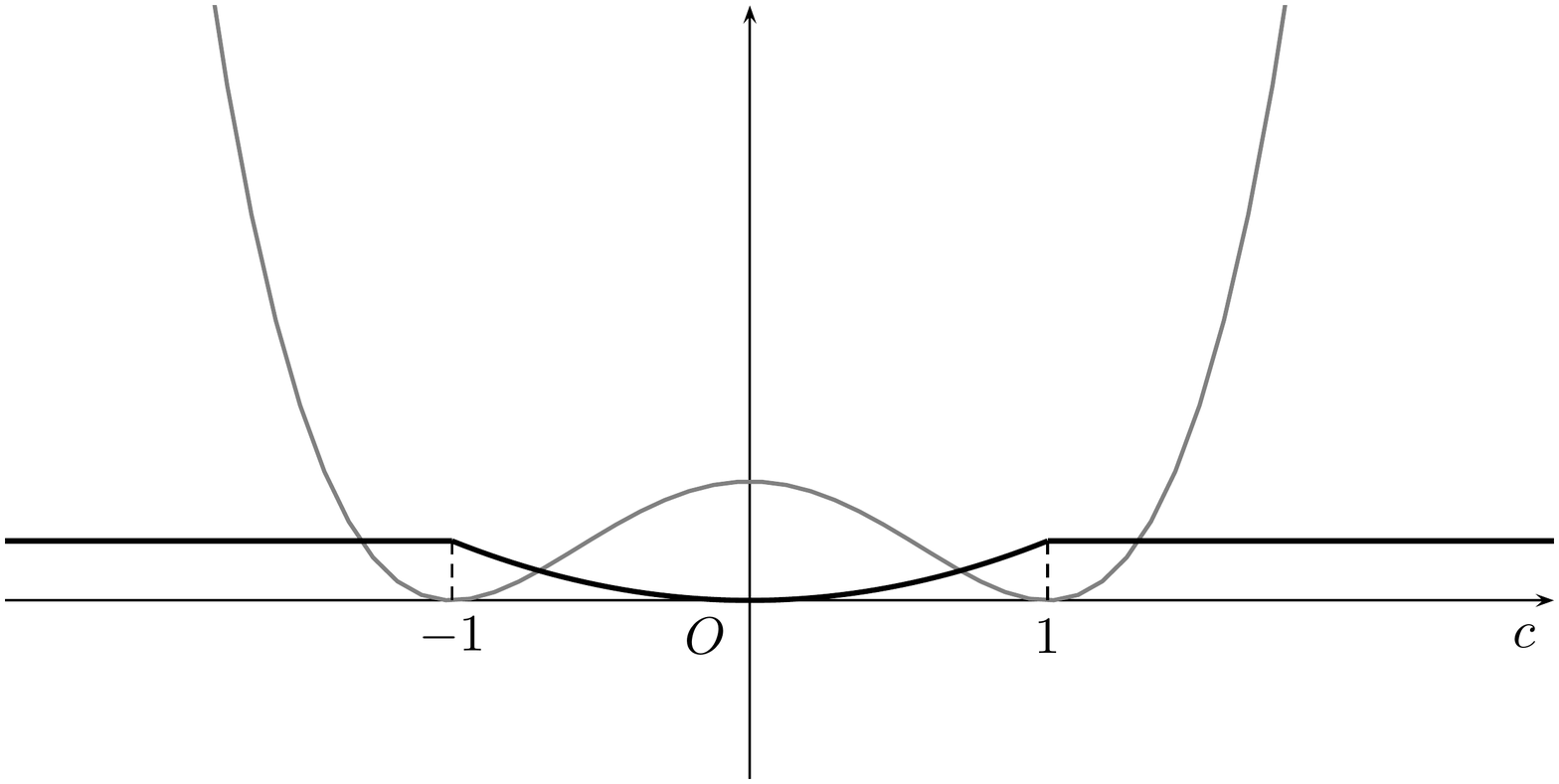}
\vskip -1cm
\caption{The plot of $F$ (gray line) and $G$ (black line).}
\end{figure}

\noindent
Hence $F\geq 0$ and $F$ vanishes only at $c=-1,1$. Moreover, by \eqref{F}--\eqref{G} we have 
\begin{equation}
F_c(c)=4c(c^{2}-1)\ ,\qquad c\in \mathbb{R}  \label{derF}
\end{equation}%
\begin{equation}
G_c(c)=\left\{ 
\begin{array}{llll}
c\qquad  & -1< c < 1 &  &  
\\[0.5em]
0\qquad  & c<-1 \, \cup \, c>1 &  & 
\end{array}%
\right.   \label{derG}
\end{equation}
We denote by $W$ the $c-$dependent part of the free energy, that is
$$
W(c) = \theta _{0}F(c) + u G(c),
\qquad
u= \theta +\displaystyle\frac{1}{\kappa _{0}}| \mathbf{q}|^{2}.
$$
The function $W$ has a unique minimum when $u \geq 4\theta_0$, while for $u < 4\theta_0$ it has two minima in $c_{\pm}$, with $|c_{\pm}|<1$ (see Fig.2). It is known \cite{CH} that the unique minimum in the potential corresponds to the situation without a miscibility gap, while in the regime with two minima there is a miscibility gap.
\begin{figure}[ht]
\includegraphics[scale=0.25]{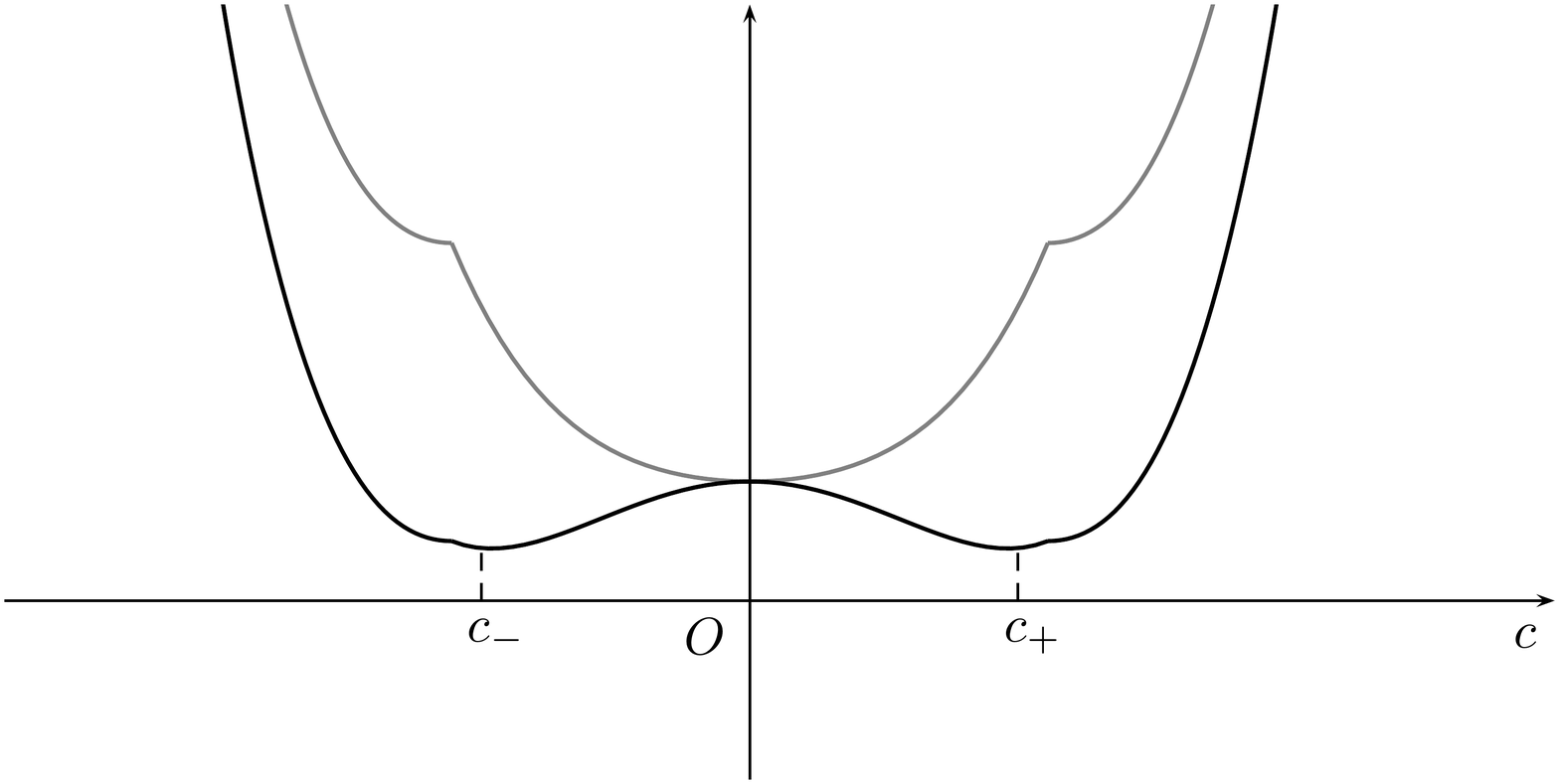}
\vskip -1cm
\caption{The plot of $W$ with $u > 4\theta_0$ (gray line) and $u < 4\theta_0$ (black line).}
\end{figure}

Finally, the mobility can be chosen as a positive function depending on $c
$. The dependence of mobility on the concentration is not new in literature:
it appeared for the first time in the original derivation of the
Cahn-Hilliard equation \cite{CH} and later other authors considered
different expressions for $M(c)$ (see for instance \cite{Barrett, elliott}).
\newline
Here the mobility $M(c)$ is taken in the form 
\begin{equation*}
M(c)=M_{0}(c^{2}-1)^{2}\ ,\qquad M_{0}>0,
\end{equation*}%
which implies that both $M$ and $\nabla M$ vanish at $c=-1,1$. Furthermore,
the mass density is such that 
\begin{equation*}
\rho (c)=\rho _{20}\ ,\quad c<-1\ ,\qquad \rho (c)=\rho _{10}\ ,\quad c>-1\ .
\end{equation*}%
In such a way $\rho $ is extended to $\mathbb{R}$.

Now we consider the initial value problem 
\begin{equation}  \label{(6.5)}
\rho(c) \dot c = \nabla\cdot[M(c)\nabla \mu(c)] \qquad c(\mathbf{x},0)=c_0(%
\mathbf{x}) \qquad \mathbf{x} \in \Omega
\end{equation}

\begin{theorem}
Let $c_0(\mathbf{x})\in [-1,1]$ for each $\mathbf{x} \in \Omega$ and assume that
\begin{equation}\label{ass}
M(c)\nabla \mu(c) \cdot \nabla \mu(c) +\gamma \,\rho(c) \nabla c\cdot \mathbf{D}\nabla c\,\geq 0.
\footnote{This inequality means that if $\mathbf{D}$ is negative definite, it is not too large relative to the first term.}
\end{equation}
Then, the solution $c(\mathbf{x},t)$ to \eqref{(6.5)} takes value in $[-1,1]$ a.e $\mathbf{x} \in \Omega$ and for each $t \in \mathbb{R}^+$.
\end{theorem}

\begin{proof}
First we prove that $c \leq 1$ a.e. $\mathbf{x} \in \Omega$ and for any $t\in \mathbb{R}^+$. We define
\begin{equation*}
c_{-}=\left\{
\begin{array}{ll}
-1 \quad & c \geq -1 \\[0.4em] 
\, c \quad & c < -1 \ .%
\end{array}
\right.
\end{equation*}
The definitions of $G$ and $c_-$ guarantee that $G(c_-)$ is a constant function, that is $G(c_-)=\frac12$, and hence
\begin{equation*}
G_{c_{-}}(c_{-})=0,
\end{equation*}
for all $c \in \mathbb{R}$.
Now, we multiply equation \eqref{(6.5)} by $\mu(c_-)$ and we integrate over $\Omega$. The divergence theorem and the boundary condition $\nabla\mu \cdot {\bf n} _{|\partial\Omega}=0$ yield 
\begin{equation}  \label{(6.7)}
\int_{\Omega}\rho\left( c\right)\dot{c}\mu \left( c_{-}\right)
dv=-\int_{\Omega }M(c)\nabla \mu \left( c\right) \cdot \nabla \mu \left(
c_{-}\right) dv.
\end{equation}
We focus our attention on the left--hand side of \eqref{(6.7)}. 
Accounting for $\nabla c_{-} \cdot \mathbf{n}=0$ at $\partial \Omega$, we obtain 
\begin{eqnarray*}
& &\int_{\Omega }\rho\left( c\right) \,\dot{c}\mu \left( c_{-}\right) dv
=
\int_{\Omega }\rho\left( c\right) \,\dot{c} 	\Big[ \theta _{0} F_{c_{-}}\left( c_{-}\right) - \frac{\gamma }{\rho( c_{-}) }\nabla \cdot \left( \rho(c_{-}) \nabla \,c_{-}\right) \Big]dv
\\
&=&
\theta _{0}\int_{\Omega }\rho \,\left( c\right) \,\dot{c}F_{c_{-}}\left(c_{-}\right) dv
+
\gamma \int_{\Omega }\rho\left( c_{-}\right) \nabla \,c_{-} \cdot  \nabla \left[
\frac{\rho\left( c\right) }{\rho \left( c_{-}\right) }\,\dot{c}\right]
dv.
\end{eqnarray*}
Following \cite{FGM}, we deduce the equality
\begin{equation*}
\begin{split}
&\int_{\Omega }\rho \left( c\right) \dot{c}\mu \left( c_{-}\right) dv
=
\\
&=
\frac{d}{dt}\int_{\Omega }\rho\left(c\right) \left[ \theta _{0}F\left(
c_{-}\right) +\frac{\gamma }{2}\left| \nabla \,c_{-}\right| ^{2}\right] dv
+
\gamma \int_{\Omega}\rho\left( c\right) \nabla \,c_{-}\cdot \mathbf{D}\nabla c_{-}\,dv.
\end{split}
\end{equation*}
Now we look at the right-hand side of \eqref{(6.7)}. Since $M(c)\nabla \mu \left( c_{-}\right)$ vanishes when $c \geq -1$, we can write 
\begin{equation*}
\int_{\Omega }M(c)\nabla \mu \left( c\right) \cdot \nabla \mu \left(
c_{-}\right) dv
= \int_{\Omega } M(c_{-}) |\nabla \mu \,\left( c_{-}\right)|^2 dv.
\end{equation*}
Collecting all the results we obtain 
\begin{equation*}
\begin{split}
&\frac{d}{dt}\int_{\Omega }\rho(c) \left[ \theta _{0}F\left(
c_{-}\right) +\frac{\gamma }{2}\left| \nabla \,c_{-}\right| ^{2}\right] dv
=
\\
&=
-\int_{\Omega }\left[ M(c_{-}) |\nabla \mu \,\left( c_{-}\right)|^2 +\gamma \,\rho\left( c\right) \nabla \,c_{-}\cdot 
\mathbf{D}\nabla \,c_{-} \right] dv.
\end{split}
\end{equation*}
Assumption \eqref{ass} allows us to conclude that
\begin{equation*}
\frac{d}{dt}\int_{\Omega }\rho(c) \left[ \theta _{0}F\left(
c_{-}\right) +\frac{\gamma }{2}\left| \nabla c_{-}\right| ^{2}\right]
dv\leq 0.
\end{equation*}
Since $|c_0({\bf x})|\leq 1$, we have $c_-({\bf x},0) =-1$ and hence
$F(c_{-}(\mathbf{x}, 0))=0$, $\nabla c_{-}(\mathbf{x}, 0)=0$. Thus, an integration over $t \in [0,T]$ yields 
\begin{equation*}
\int_{\Omega }\rho(c) \left[ \theta _{0}F\left( c_{-}\right) +
\frac{\gamma }{2}\left| \nabla c_{-}\right| ^{2}\right] \mid _{t=T}dv\leq
0 \qquad T\in \mathbb{R}^+
\end{equation*}
which implies that $F(c_{-}(\mathbf{x}, T))=0$, $\nabla c_{-}(\mathbf{x}, T)=0$.
Since $c_{-} \leq -1$ and F is non-negative and vanishes only at $-1$, $1$
then it follows that $c_{-}(\mathbf{x}, T) = -1$, namely 
\begin{equation*}
c(\mathbf{x}, T) \geq -1.
\end{equation*}

One can easily show that $c(\mathbf{x}, T) \leq 1$ by defining 
\begin{equation*}
c_{+}=\left\{ 
\begin{array}{ll}
1 \quad & c \leq 1 \\[0.4em] 
c \quad & c > 1
\end{array}
\right.
\end{equation*}
and repeating step by step the procedure adopted for $c_{-}$.

In conclusion 
\begin{equation*}
c(\mathbf{x},T)\in [-1,1] \quad \mathbf{x}\in \Omega, T \in \mathbb{R}^+
\end{equation*}
and the theorem is proved.
\end{proof}


\subsection*{Acknowledgment}
The first two authors have been partially supported by G.N.F.M. - I.N.D.A.M. through the projects for young researchers ``Mathematical models for multiphase materials''.

\end{document}